\documentclass{llncs}
\usepackage{amsmath}

\newcommand{\Oh}[1]
  {\ensuremath{\mathcal{O}\!\left( {#1} \right)}}
\newcommand{\LCS}
  {\ensuremath{\mathsf{LCS}}}

\newcommand{\BWD}
  {\ensuremath{\mathsf{BWD}}}
\newcommand{\BWT}
  {\ensuremath{\mathsf{BWT}}}
\newcommand{\rank}
  {\ensuremath{\mathsf{rank}}}
\newcommand{\select}
  {\ensuremath{\mathsf{select}}}
\long\def\ignore#1{\vskip 0pt}

\begin{document}

\title{Relative FM-indexes}

\ignore{\author{Djamal Bellazzougui\inst{1} \and Travis Gagie\inst{1} \and
Simon Gog\inst{2} Giovanni Manzini\inst{3} \and Jouni Sir\'en\inst{4}}
\institute{University of Helsinki \and University of Melbourne \and
University of Eastern Piedmont \and University of Chile}}

\maketitle

\begin{abstract}
Intuitively, if two strings $S_1$ and $S_2$ are sufficiently similar and we already have an FM-index for $S_1$ then, by storing a little extra information, we should be able to reuse parts of that index in an FM-index for $S_2$.  We formalize this intuition and show that it can lead to significant space savings in practice, as well as to some interesting theoretical problems.
\end{abstract}

\section{Introduction} \label{sec:introduction}

FM-indexes~\cite{FM05} are core components in most modern DNA aligners
(e.g.,~\cite{LTPS09,LD09,LYLLYKW09}) and have thus played an important role
in the genomics revolution.  Medical researchers are now producing databases
of hundreds or even thousands of human genomes, so bioinformatics researchers
are working to improve FM-indexes' compression of sets of nearly duplicate
strings.  As far as we know, however, the solutions proposed so far
(e.g.,~\cite{FGHP14,MNSV10}) index the concatenation of the genomes, so we
can search the whole database easily but searching only in one specified
genome is more difficult.  In this paper we consider how to index each of the
genomes individually while still using reasonable space and query time.

Our intuition is that if two strings $S_1$ and $S_2$ are sufficiently similar
and we already have an FM-index for $S_1$ then, by storing a little extra
information, we should be able to reuse parts of that index in an FM-index
for $S_2$. More specifically, it seems $S_1$'s and $S_2$'s Burrows-Wheelers
Transforms~\cite{BW94} (BWTs) should also be fairly similar. Since BWTs are
the main component of FM-indexes, it is natural to try to take advantage of
such similarity to build an index for $S_2$ that ``reuses'' information
already available in $S_1$'s FM-index.

Among the many possible similarities one can find and exploit in the BWTs, in
this paper we consider the longest common subsequence (LCS). The BWT sorts
the characters of a string into the lexicographic order of the suffixes
following those characters.  For example, if
$$
S_1 = \mathsf{AAGTTGAGAGTGAGT},\qquad
S_2 = \mathsf{AGAGAGTCGAAGTT};
$$
then
$$
\BWT (S_1) = \mathsf{TGGGATTAAAAGTGG},\qquad
\BWT (S_2) = \mathsf{TGGGATCAAAATGG};
$$
whose LCS {\sf TGGGATAAAATGG} is nearly as long as either BWT. Note that in
this example $\LCS(S_1,S_2)=\mathsf{AGAGAGTGAGT}$ is shorter than $\LCS(\BWT
(S_1),\BWT (S_2))$.

We introduce the concept of {\em BW-distance} \(\BWD (S_1, S_2)\) between
$S_1$ and $S_2$ defined as $|S_1| + |S_2| - 2 |\LCS(\BWT (S_1),\BWT (S_2))|$.
Note that this coincides with the edit distance between $\BWT (S_1)$ and
$\BWT (S_2)$ when only insertions and deletions are allowed. We prove that,
if we are willing to tolerate a slight increase in query times, we can build
an index for $S_2$ using an unmodified FM-index for $S_1$ and additional data
structures whose total space in words is asymptotically bounded by $\BWD
(S_1, S_2)$ (Theorem~\ref{theo:counting}).

This first result is the starting point for our investigation as it generates
many challenging issues. First, since we are interested in indexing whole
genomes, we observe that finding the LCS of strings
whose length is of the order of billions is outside the capabilities of most
computers. Thus, in Section~\ref{sec:simon&jouni} we show how to approximate
the LCS of two BWTs, using combinatorial properties of
the BWT to align the sequences. In the same section we also discuss and test
several practical alternatives for building the index for $S_2$ given the one
for $S_1$ and we analyze their time/space trade-offs.

If one needs an index not only for counting queries but also for locating and
extracting, we must enrich it with suffix array (SA) samples. Such samples
usually take significantly less space than the main index. However, we may
still want to take advantage of the similarities between $S_1$ and $S_2$ to
``reuse'' SA samples from $S_1$ for $S_2$'s index. In
Section~\ref{sec:djamal&giovanni} we show that this is indeed possible if,
instead of considering the LCS between the BWTs, we use a common subsequence
with the additional constraint of being {\em BWT-invariant}
(Theorem~\ref{theo:locating}). This result motivates the problem of finding
the longest BWT-invariant subsequence, which unfortunately turns out to be
NP-hard (Theorem~\ref{theo:np}). We therefore devise a heuristic to find a
``long'' BWT-invariant subsequence in $\Oh{|S_1|\log|S_1|}$ time.

We have tested our approach in practice by building an FM-index for the genome
of a Han Chinese individual, ``reusing'' an FM-index of the human reference
genome. The Han genome is about 3.0 billion base pairs, the reference is
about 3.1 billion base pairs and we found a common subsequence of about 2.9
billion base pairs. A standard implementation of a stand-alone FM-index for
the Han genome takes 628~MB or 1090~MB, depending on encoding, while our
index uses only 256~MB or 288~MB on top of the index for the reference.  On
the other hand, queries to our index take about 9.5 or 4.5 times longer.
Since our index is compressed relative to the underlying index for the
reference, we call it a relative FM-index.

\ignore{In this paper we are mainly concerned with the practical applications
of our observation.  In Section~\ref{sec:simon&jouni} we describe how we
built an FM-index for the genome of a Han Chinese individual, reusing a
\rank\ data structure over the BWT of the human reference genome. The Han
genome is about 3.0 billion base pairs, the reference is about 3.1 billion
base pairs and we found a common subsequence of about 2.9 billion base pairs.
A standard implementation of a stand-alone FM-index for the Han genome takes
628~MB or 1090~MB, depending on encoding, while our index uses only 256~MB or
288~MB on top of the index for the reference.  On the other hand, queries to
our index take about 9.5 or 4.5 times longer.  Since our index is compressed
relative to the underlying index for the reference, we call it a relative
FM-index.

The \rank\ data structure over the BWT normally dominates the space usage of
an FM-index because the other components all have sublinear size.  If we
store enough relative FM-indexes and our space savings for the \rank\ data
structures are good enough, however, those other components could become a
concern.  In particular, the suffix array (SA) sample used for locating and
extracting usually takes an only slightly sublinear number of bits, unless
those operations are slow.  (Bitvectors with constant-time queries can also
take quite a lot of space, but that can be reduced by increasing the query
times.)  We have developed (but not yet tested) a way to reuse an SA sample
--- see the appendix for an overview --- but it involves finding a long
common subsequence in \(\BWT (S_1)\) and \(\BWT (S_2)\) such that the
relative order of its characters in $S_1$ and $S_2$ is the same.  We show in
Section~\ref{sec:djamal&giovanni} that finding the longest such subsequence
is NP-hard, but we also give heuristics that seem to work well in practice.

}

\section{Review of the FM-index structure} \label{sec:review}

The core component of an FM-index for a string \(S [1..n]\) is a data
structure supporting rank queries on the Burrows-Wheeler
Transform \(\BWT (S)\) of $S$.  This transform permutes the
characters in $S$ such that \(S [i]\) comes before \(S [j]\) in \(\BWT (S)\)
if \(S [i + 1..n]\) is lexicographically less than \(S [j + 1..n]\).

If the lexicographic range of suffixes of $S$ starting with $\beta$ is
\([i..j]\), then the range of suffixes starting with \(a \beta\) is
\begin{multline*}
\left[ \BWT (S).\rank_a (i - 1) + 1 + \sum_{a' \prec a} S.\rank_{a'} (n).. \right.\\[-2ex]
\vspace{-1ex} \left. \BWT (S).\rank_a (j) + \sum_{a' \prec a} S.\rank_{a'} (n) \right]
\end{multline*}
It follows that, if we have precomputed an array storing \(\sum_{a' \prec a}
S.\rank_{a'} (n)\) for each distinct character $a$ (i.e., the number of
characters in $S$ less than $a$), then we can find the range of suffixes
starting with a pattern \(P [1..m]\) --- and, thus, count its occurrences ---
using $\Oh{m}$ rank queries.

If the position of \(S [i]\) in \(\BWT (S)\) is $j$, then the position of \(S
[i - 1]\) is
\[\BWT(S).\rank_{S [i]} (j) + \sum_{a \prec S [i]} \BWT (S).\rank_{a} (n)\,.\]
It follows that, if we have also precomputed a dictionary storing the
position of every $r$th character of $S$ in \(\BWT (S)\) with its position in
$S$ as satellite information, then we can find a character's position in $S$
from its position in \(\BWT (S)\) using $\Oh{r}$ rank and membership queries.
Therefore, once we know the lexicographic range of suffixes starting with
$P$, we can locate each of its occurrences using $\Oh{r}$ rank queries.

Finally, if we have also precomputed an array storing the position of every
$r$th character of $S$ in \(\BWT (S)\), in order of appearance in $S$, then
given $i$ and $j$, we can extract \(S [i..j]\) using $\Oh{r + j - i}$ rank
queries.


\section{BW-distance and relative FM-indices} \label{sec:BWdiv}

Given two strings \(S_1 [1..n_1]\) and \(S_2 [1..n_2]\) we define the {\em
BW-distance\/} \(\BWD (S_1, S_2)\) between $S_1$ and $S_2$ as
\begin{equation}\label{eq:BWDdef}
\BWD(S_1, S_2) = n_1 + n_2 - 2 |\LCS(\BWT (S_1),\BWT (S_2))|.
\end{equation}
Note that the BW-distance is nothing but the edit distance between
$\BWT(S_1)$ and $\BWT(S_2)$ when only insertions and deletions are
allowed~\cite{Myers86} (also known as the shortest edit script or indel distance),
and is thus at most twice their normal edit distance. We now show
how to support counting queries on $S_2$ using an FM-index for $S_1$ and some
auxiliary data structures taking $\Oh{\BWD (S_1, S_2)}$ words of space.
Specifically, we consider how we can support rank queries on \(\BWT (S_2)\)
and partial-sum queries on the distinct characters' frequencies.

Let $C$ denote a LCS of $\BWT (S_1)$ and $\BWT(S_2)$ with $|C|=m$. Let $C=c_1
\cdots c_m$, and for $i=1,\ldots,m$, let $\alpha_i$ (resp. $\beta_i$) be the
position of $c_i$ in $\BWT(S_1)$ (resp. $\BWT(S_s)$) with $\alpha_1 < \cdots
< \alpha_{m}$ (resp. $\beta_1 < \cdots < \beta_{m})$. Define
\begin{itemize}
\item bitvector \(B_1 [1..n_1]\) with 0s in positions  $\alpha_1, \ldots,
    \alpha_m$,
\item bitvector \(B_2 [1..n_2]\) with 0s in positions of $\beta_1,
    \ldots, \beta_m$,
\item subsequence $D_1$ of \(\BWT (S_1)\) marked by 1s in $B_1$; $D_1$ is
    the complement of $C$ in $\BWT (S_1)$,
\item subsequence $D_2$ of \(\BWT (S_2)\) marked by 1s in $B_2$; $D_2$ is
    the complement of $C$ in $\BWT (S_2)$.
\end{itemize}

We claim that if we can support fast \rank\ queries on \(\BWT (S_1)\), $B_1$,
$B_2$, $D_1$ and $D_2$ and fast $\select_0$ queries on $B_1$, then we can
support fast \rank\ queries on \(\BWT (S_2)\). To see why, notice that
\begin{align*}
\BWT (S_2).\rank_X (i)
& = C.\rank_X (B_2.\rank_0 (i))\\
& \quad + D_2.\rank_X (B_2.\rank_1 (i))
\end{align*}
and, by the same reasoning,
\begin{align*}
C.\rank_X (j)
& = \BWT (S_1).\rank_X (B_1.\select_0 (j))\\
& \quad - D_1.\rank_X (B_1.\rank_1 (B_1.\select_0 (j)))\,.
\end{align*}
Therefore,
\begin{align*}
\BWT (S_2).\rank_X (i)
& = \BWT (S_1).\rank_X (k)\\
& \quad - D_1.\rank_X (B_1.\rank_1 (k))\\
& \quad + D_2.\rank_X (B_2.\rank_1 (i))
\end{align*}
where \(k = B_1.\select_0 (B_2.\rank_0 (i))\).

For example, for the strings
$$
S_1 = \mathsf{AAGTTGAGAGTGAGT},\qquad
S_2 = \mathsf{AGAGAGTCGAAGTT};
$$
it is
$$
\BWT (S_1) = \mathsf{TGGGATTAAAAGTGG},\qquad
\BWT (S_2) = \mathsf{TGGGATCAAAATGG};
$$
and $\LCS(\BWT (S_1),\BWT (S_2)) = \mathsf{TCTCGTAAAAGG}$. Hence
\begin{align*}
B_1  & = 0001000000000111  & D_1  &= \mathsf{GTGC}\\
B_2  & = 010000000001010   & D_2  & = \mathsf{GCC}.
\end{align*}
Suppose we want to compute $\BWT (S_2).\rank_\mathsf{C}$. It is
\(B_1.\select_0 (B_2.\rank_0 (13)) = 12\), so
\begin{align*}
\BWT (S_2).\rank_\mathsf{C} (13)
& = \BWT (S_1).\rank_\mathsf{C} (12)\; -\; D_1.\rank_\mathsf{C} (B_1.\rank_1 (12))
\\
& \quad + D_2.\rank_\mathsf{C} (B_2.\rank_1 (13)) \;=\; 3.
\end{align*}

Observing that the number of 1s in $B_1$ and $B_2$ is \(\Oh{\max (n_1, n_2) -
\ell} = \Oh{\BWD (S_1, S_2)}\), we can store data structures for $B_1$,
$B_2$, $D_1$ and $D_2$ in $\Oh{\BWD (S_1, S_2)}$ space such that the desired
\rank/\select\ queries take $\Oh{\log \BWD (S_1, S_2)}$ time.

The only other component required for an FM-index for $S_2$ for counting, is a data
structure for computing  \(\sum_{a' \prec a} S_2.\rank_{a'} (n)\) for each
distinct character~$a$. Notice that \(\BWD (S_1, S_2)\) is at least the
number of distinct characters whose frequencies in $S_1$ and $S_2$ differ. It
follows that in $\Oh{\BWD (S_1, S_2)}$ space we can store
\begin{itemize}
\item a $\Oh{\log \BWD (S_1, S_2)}$-time predecessor data structure
    storing those distinct characters,
\item an array storing \(\sum_{a' \prec a} S_2.\rank_{a'} (n_2)\) for
    each such distinct character $a$.
\end{itemize}
For any distinct character $b$, we can find the preceding distinct character
$a$ whose frequencies in $S_1$ and $S_2$ differ and compute
\[\sum_{a' \prec b} S_2.\rank_{a'} (n_2) = \sum_{a' \prec b} S_1.\rank_{a'} (n_1) - \sum_{a' \prec a} S_1.\rank_{a'} (n_1) + \sum_{a' \prec a} S_2.\rank_{a'} (n_2)\]
using $\Oh{\log \BWD (S_1, S_2)}$ time. Summing up:

\begin{theorem} \label{theo:counting}
If we already have an FM-index for $S_1$, we can store a relative FM-index
for $S_2$ using $\Oh{\BWD (S_1, S_2)}$ words of extra space. Counting queries
on the relative FM-index take time an $\Oh{\log \BWD (S_1, S_2)}$ factor
larger than on $S_1$.
\end{theorem}

In Section~\ref{sec:djamal&giovanni} we show how to build a relative FMindex
supporting also locating and extracting.


\subsection{A practical implementation} \label{sec:simon&jouni}

A longest common sequence of $\BWT (S_1)$ and $\BWT (S_2)$ can be computed in
$\Oh{n_1 n_2 /w}$ time, where $w$ is the word size~\cite{Myers99}. Since we
are mainly interested in strings with a small BW-distance, a better
alternative could be the algorithms whose running times are bounded by the
number of differences between the input sequences (see
eg~\cite{LandauVN86,Myers86}). Unfortunately none of these algorithms is
really practical when working with such very large files as the complete
genomes we considered in our tests. Hence, to make the construction of a
relative FM-index practical, we approximate the LCS of the two
Burrows-Wheeler transforms, using the combinatorial properties of the BWT to
align the sequences.

Let $S_{1}$ be a random string of length $n$ over alphabet $\Sigma$ of size
$\sigma$, and let string $S_{2}$ differ from it by $s$ insertions, deletions,
and substitutions. In the expected case, the edit operations move $O(s
\log_{\sigma} n)$ suffixes in lexicographic order, and change the preceding
characters for $O(s)$ suffixes~\cite{MNSV10}. If we remove the characters
corresponding to those suffixes from $\BWT(S_{1})$ and $\BWT(S_{2})$, we have
a common subsequence of length $n - O(s \log_{\sigma} n)$ in the expected
case.

Assume that we have partitioned the BWTs according to the first $k$
characters of the suffixes, for $k \ge 0$. For all $x \in \Sigma^{k}$, let
$\BWT_{x}(S_{1})$ and $\BWT_{x}(S_{2})$ be the substrings of the BWTs
corresponding to the suffixes starting with $x$. If we remove the suffixes
affected by the edit operations, as well as the suffixes where string $x$
covers an edit, we have a common subsequence $\BWT_{x}'$ of $\BWT_{x}(S_{1})$
and $\BWT_{x}(S_{2})$. If we concatenate the sequences $\BWT_{x}'$ for all
$x$, we get a common subsequence of $\BWT(S_{1})$ and $\BWT(S_{2})$ of length
$n - O(s (k + \log_{\sigma} n))$ in the expected case. This suggests that we
can find a long common subsequence of $\BWT(S_{1})$ and $\BWT(S_{2})$ by
partitioning the BWTs, finding an LCS for each partition, and concatenating the
results.

In practice, we partition the BWTs by variable-length strings. We use
backward searching on the BWTs to traverse the suffix trees of $S_{1}$ and
$S_{2}$, selecting a partition when either the length of $\BWT_{x}(S_{1})$ or
$\BWT_{x}(S_{2})$ is at most $1024$, or the length of the pattern $x$ reaches
$32$. For each partition, we use the greedy LCS algorithm~\cite{Myers86} to
find the longest common subsequence of that partition. To avoid hard cases,
we stop the greedy algorithm if it would need diagonals beyond $\pm 50000$,
and match only the most common characters for that partition. We also predict
in advance the common cases where this happens (the difference of the lengths
of $\BWT_{x}(S_{1})$ and $\BWT_{x}(S_{2})$ is over $50000$, or $x = N^{32}$
for DNA sequences), and match the most common characters in that partition
directly.

We implemented the counting structure of the relative FM-index using the SDSL
library~\cite{Gog2014b}, and compared its performance to a regular FM-index.
To encode the BWTs and sequences $D_{1}$ and $D_{2}$, we used Huffman-shaped
wavelet trees with either plain or entropy-compressed (RRR)~\cite{Raman2007}
bitvectors. We chose entropy-compressed bitvectors for marking the positions
of the LCS in $\BWT(S_{1})$ and $\BWT(S_{2})$.

The implementation was written in C++ and compiled on g++ version 4.7.3. We
used a system with 32 gigabyes of memory and two quad-core 2.53 GHz Intel
Xeon E5540 processors, running Ubuntu 12.04 with Linux kernel 3.2.0. Only one
CPU core was used in the experiments.

For our experiments, we used the 1000 Genomes Project assembly of the human
reference genome as the reference sequence $S_{1}$.\footnote{GRCh37,
\url{ftp://ftp-trace.ncbi.nih.gov/1000genomes/ftp/technical/reference/}} As
sequence $S_{2}$, we used the genome of a Han Chinese individual from the
YanHuang
project.\footnote{\url{ftp://public.genomics.org.cn/BGI/yanhuang/fa/}} The
lengths of the sequences were 3.10 billion bases and 3.00 billion bases,
respectively, and our algorithm found a common subsequence of 2.93 billion
bases. As our pattern set, we used 10 million reads of length 56. Almost 4.20
million reads had exact matches in sequence $S_{2}$, with a total of 99.7
million occurrences. The results of the experiments can be seen in
Table~\ref{table:experiments}.

\begin{table}[t]
\centering
\caption{Experiments with human genomes. Bitvector used in the wavelet tree; time and space requirements for building the relative FM-index; time required for counting queries and index size for a regular and a relative FM-index; the performance of the relative FM-index compared to the regular index. The query times are averages over five runs.}\label{table:experiments}
\begin{tabular}{ccccccccccccc}
\hline
\noalign{\smallskip}
 &\phantom{00} & \multicolumn{2}{c}{Construction} &\phantom{00} & \multicolumn{2}{c}{Regular}
 &\phantom{00} & \multicolumn{2}{c}{Relative} &\phantom{00} & \multicolumn{2}{c}{Rel vs.~Reg} \\
Bitvector & & Time & Space & & Time & Size & & Time & Size & & Time & Size \\
\noalign{\smallskip}
\hline
\noalign{\smallskip}
Plain & &  762\:s & 9124\:MB & & 146\:s & 1090\:MB & & 1392\:s & 288\:MB & & 954\% & \ 26\% \\
\noalign{\smallskip}
RRR   & & 6022\:s & 7823\:MB & & 667\:s &  628\:MB & & 3022\:s & 256\:MB & & 453\% & \ 41\% \\
\noalign{\smallskip}
\hline
\end{tabular}
\end{table}

With plain bitvectors in the wavelet tree, the relative FM-index was 9.5 times slower than a regular FM-index, while requiring a quarter of the space. With entropy-compressed bitvectors, the relative index was 4.5 times slower and required 41\% of the space. Comparing the relative FM-index using plain bitvectors to the regular index using entropy-compressed bitvectors, we see that the relative index is 2.1 times slower, while taking 46\% of the space.

Bitvectors $B_{1}$ and $B_{2}$ took 70\% to 80\% of the total size of the
relative index. We tried to encode them as sparse
bitvectors~\cite{Okanohara2007}, but the result was slightly larger and
clearly slower than with entropy-compressed bitvectors. By our estimates,
run-length encoded bitvectors would have taken slightly more space than
sparse vectors. Hybrid bitvectors using different encodings for different
parts of the bitvector~\cite{Kaerkkaeinen2014} could improve compression, but
the existing implementation does not work with vectors longer than $2^{31}$
bits.

\section{Relative FM-indices supporting locating and extracting}
\label{sec:djamal&giovanni}

As mentioned in Section~\ref{sec:review}, an FM-index for $S_1$ usually has
an SA sample that takes an only slightly sublinear number of bits.  This
sample has two parts: the first consists of a bitvector $R$ with 1s marking
the positions in \(\BWT (S_1)\) of every $r$th character in $S_1$, and an
array $A$ storing a mapping from the ranks of those characters' positions in
\(\BWT (S_1)\) to their positions in $S_1$; the second is an array storing a
mapping from the ranks of those characters' positions in $S$ to their
positions in \(\BWT (S_1)\).  With these, given the position of a sampled
character in \(\BWT (S_1)\), we can find its position in $S_1$, and vice
versa.

These parts are used for locating and extracting queries, respectively, and
the worst-case query times are proportional to $r$.  On the other hand, the
size of the sample in words is proportional to the length of $S$ divided by
$r$.  For details on how the sample works, we direct the reader to the full
description of FM-indexes~\cite{FM05}.  We note only that if we sample
irregularly, then the worst-case query times for locating and extracting are
proportional to the maximum distance in $S$ between two consecutive sampled
characters.  We leave consideration of extracting for the full version of the
paper --- it is nearly symmetric to locating --- so we do not discuss the
second part of the sample here.

Let
\(G = S_1[i_1]\;\cdots,\;S_1[i_\ell]\)
denote a length-$\ell$ common subsequence of $S_1$ and $S_2$ (not their
BWTs). That is, we have $i_1< \;\cdots\; < i_\ell$ and there exists $j_1 <
\;\cdots\; < j_\ell$ such that
$$
S_1[i_1]=S_2[j_1],\;\ldots,\;S_1[i_\ell]=S_2[j_\ell].
$$
Since there is a one-to-one correspondence between a text and its BWT, we can
define the indexes $v_1, \ldots, v_\ell$ (resp. $w_1,\ldots,w_\ell$) such
that for $k=1,\ldots,\ell$, $\BWT(S_1)[v_k]$ is the character corresponding
to $S_1[i_k]$ (resp. $\BWT(S_2)[w_k]$ is the character corresponding to
$S_2[j_k]$). We say that the common subsequence $G$ is {\em BWT-invariant} if
there exists a permutation $\pi: \{1,\ldots,\ell\} \rightarrow
\{1,\ldots,\ell\}$ such that we have simultaneously
\begin{equation}\label{eq:bwtinv}
v_{\pi(1)} < v_{\pi(2)} < \cdots < v_{\pi(\ell)}, \quad\mbox{and}\quad
w_{\pi(1)} < w_{\pi(2)} < \cdots < w_{\pi(\ell)}.
\end{equation}
In other words, when we go from the texts to the BWTs the elements of $G$ are
permuted in the same way in $S_1$ and $S_2$.

\ignore{
\begin{defintion}\label{def:G}
Given the indexes $i_1,\ldots,i_\ell$, $j_1,\ldots,j_\ell$,
$v_1,\ldots,v_\ell$ and $w_1,\ldots,w_\ell$ defined as above, \qed
\end{definition}}

An immediate consequence of~\eqref{eq:bwtinv} is that the sequence
$$
G' = \BWT(S_1)[v_{\pi(1)}]\,\BWT(S_1)[v_{\pi(2)}]\,\cdots\,\BWT(S_1)[v_{\pi(\ell)}]
$$
is a common subsequence of $\BWT(S_1)$ and $\BWT(S_2)$. We can therefore
generalize~\eqref{eq:BWDdef} and define
$$
\BWD_G(S_1, S_2) = \max(n_1,n_2) - |G|
$$
and repeat the construction of Theorem~\ref{theo:counting} with $\BWD$
replaced by $\BWD_G$. However, since $G$ is BWT-invariant it is now possible
to reuse the the SA samples from $S_1$ relative to positions in $G$ for the
string $S_2$ provided that we have
\begin{itemize}
\item bitvector \(M_1 [1..n_1]\) with 0s in positions  $i_1, \ldots, i_\ell$,
    supporting fast \rank\ queries,
\item bitvector \(M_2 [1..n_2]\) with 0s in positions of $j_1, \ldots,
    j_\ell$, supporting fast $\select_0$ queries;
\end{itemize}
proof idea in the appendix, complete proof in the full paper. Summing up,
we have:

\begin{theorem} \label{theo:locating}
For any BWT-invariant subsequence $G$, if we already have an FM-index for
$S_1$, then we can store $\Oh{\BWD_G (S_1, S_2)}$ extra space such that the
time bounds for locating and extracting queries on $S_2$ are an $\Oh{\log
\BWD_G (S_1, S_2)}$ factor larger than on $S_1$.
\end{theorem}

In view of the above theorem, it is certainly desirable to find the longest
common subsequence of $S_1$ and $S_2$ which is BWT-invariant. Unfortunately,
this problem is NP-hard as shown by the following result.

\begin{theorem}\label{theo:np}
It is NP-complete to determine whether there is an LCS of $S_1$ and $S_2$
which is BWT-invariant, even when the strings are over a ternary alphabet.
\end{theorem}

\begin{proof}
Clearly we can check in polynomial time whether a given subsequence of $S_1$
and $S_2$ has this property, so the problem is in NP.  To show that it is
NP-complete, we reduce from the NP-complete problem of permutation pattern
matching~\cite{BBL98}, for which we are given two permutations $\pi_1$ and
$\pi_2$ over $n$ and \(m \leq n\) elements, respectively, and asked to
determine whether there is a subsequence of $\pi_1$ of length $m$ such that
the relative order of the elements in that subsequence is the same as the
relative order of the elements in $\pi_2$.  For example, if \(\pi_1 = 6, 3,
2, 1, 4, 5\) and \(\pi_2 = 4, 2, 1, 3\), then \(6, 2, 1, 4\) is such a
subsequence. Specifically, we set
\begin{align*}
S_1 & = \mathsf{A B^{\pi_1 [1]} A B^{\pi_1 [2]} \cdots A B^{\pi_1 [n]}}\\
S_2 & = \mathsf{A C^{\pi_2 [1]} A C^{\pi_2 [2]} \cdots A C^{\pi_2 [m]}}\,,
\end{align*}
so the unique LCS of $S_1$ and $S_2$ is $\mathsf{A}^m$.  For our example,
\begin{align*}
S_1 & = \mathsf{A B^6 A B^3 A B^2 A B A B^5} = \mathsf{A B B B B B B A B B B A B B A B A B B B B B}\\
S_2 & = \mathsf{A C^4 A C^2 A C A C^3} = \mathsf{A C C C C A C C A C A C C C}\,.
\end{align*}
The BWT sorts the $m$ copies of {\sf A} in $S_2$ according to $\pi_2$ and
sorts any subsequence of $m$ copies of {\sf A} in $S_1$ according to the
corresponding subsequence of $\pi_1$.  Therefore, there is an LCS of $S_1$
and $S_2$ such that the relative order of its characters is \(\BWT (S_1)\)
and \(\BWT (S_2)\) is the same, if and only if there is a subsequence of
$\pi_1$ of length $m$ such that the relative order of the elements in that
subsequence is the same as the relative order of the elements in $\pi_2$.\qed
\end{proof}

In view of the above result, for large inputs we cannot expect to find the
longest possible BWT-invariant subsequence, so, as for the LCS, we have
devised the following fast heuristic for computing a ``long'' BWT-invariant
subsequence.

We first compute the suffix array $SA_{12}$ for the concatenation $S_1\#S_2$
and we use it to define the array $A$ of size $n_1 \times 2$ as follows
\begin{itemize}
\item $A[i][1] = j\quad$ iff $S_1[i] = S_2[j]$ and suffix $S_2[j+1,n_2]$
    immediately follows suffix $S_1[i+1,n_1]$ in $SA_{12}$. If no such
    $j$ exists $A[i][1]$ is undefined.
\item $A[i][2] = j\quad$ iff $S_1[i] = S_2[j]$ and suffix $S_2[j+1,n_2]$
    is the lexicographically largest suffix of $S_2$ preceding suffix
    $S_1[i+1,n_1]$ in $SA_{12}$. In no such $j$ exists $A[i][2]$ is
    undefined.
\end{itemize}
Next, we compute the longest subsequence $1 \leq i_1 < i_2 < \cdots < i_\ell
\leq n_1$ such that there exist $b_1, \ldots, b_\ell$, with $b_k \in \{1,2\}$
and the sequence
$$
A[i_1][b_1] < A[i_2][b_2] < \cdots < A[i_\ell][b_\ell]
$$
is the longest possible (every $A[i_k][b_k]$ must be defined). The values
$i_1, \ldots, i_\ell$ and $b_1, \ldots, b_\ell$ can be computed in
$\Oh{n_1\log n_1}$ time using a straightforward modification of the dynamic
programming algorithm for the longest increasing subsequence. Setting, for
$k=1,\ldots,\ell$, $j_k = A[i_k][b_k]$ we get that
$$
G \;= \; S_1[i_1] S_1[i_2] \cdots S_1[i_\ell] \; = \; S_2[j_1] S_2[j_2] \cdots S_2[j_\ell]
$$
is a common subsequence of $S_1$ and $S_2$.

\begin{lemma}\label{lemma:bwtinv}
The subsequence $G$ is BWT-invariant.
\end{lemma}

\begin{proof}
Let $v_1, \ldots, v_\ell$ (resp. $w_1,\ldots,w_\ell$) such that for
$k=1,\ldots,\ell$, $\BWT(S_1)[v_k]$ is the character corresponding to
$S_1[i_k]$ (resp. $\BWT(S_2)[w_k]$ corresponds to $S_2[j_k]$). It suffices to
prove that for any pair $h,k$, with $1 \leq h,k \leq \ell$,  the inequality
$v_h < v_k$ implies $w_h < w_k$. Let $\prec$ denote the lexicographic order.
By construction, and by the properties of the BWT, we have $v_h < v_k$ iff
the suffix $S_1[i_h+1,n_1] \prec S_1[i_k+1,n_1]$ and we must prove that this
implies $S_2[j_h+1,n_2] \prec S_2[j_k+1,n_2]$.

Since $j_h = A[i_h][b_h]$ and $j_k = A[i_k][b_k]$, the proof follows
considering the four possible cases: $b_h =1,2$ and $b_k=1,2$. We consider
the case $b_h=1$, $b_k=2$ leaving the others to the reader. If $j_h =
A[i_h][1]$ and $j_k = A[i_k][2]$ then $S_2[j_h+1,n_2]$ immediately follows
$S_1[i_h+1,n_1]$ in $SA_{12}$. At same time $S_2[j_k+1,n_2]$ precedes
$S_1[i_h+1,n_1]$ but there are no other suffixes from $S_2$ between them.
Since $j_h \neq j_k$ the only possible ordering of the suffixes in $SA_{12}$
is
$$
S_1[i_h+1,n_1] \;\prec\; S_2[j_h+1,n_2] \;\prec\; S_2[j_k+1,n_2] \;\prec\; S_1[i_k+1,n_1]
$$
implying $S_2[j_h+1,n_2] \prec S_2[j_k+1,n_2]$ as claimed.\qed
\end{proof}

To evaluate whether the subsequence $G$ derived from the above procedure is
still able to capture the similarity between $S_1$ and $S_2$, we have
compared the length of $G$ with the \LCS\ length for pairs of {\it
S.{}cerevisiae\/} genomes from the Saccharomyces Genome Resequencing
Project.\footnote{https://www.sanger.ac.uk/research/
projects/genomeinformatics/sgrp.html} In particular we compared the 273614N
sequence with sequences 322134S, 378604X, BC187, and DBVPG1106. For each
sequence we report in Table~\ref{table:G} the ratio between the length of $G$
and $\LCS(\BWT(S_1),\BWT(S_2))$ and the length of sequence 273614N (roughly
11.9 MB). We see that in all cases more than 85\% of BWT positions are in $G$
which roughly indicates that more than 85\% of the SA samples from 273614N
could be reused as SA samples for the other sequences.

\begin{table}[t]
\centering \caption{Comparison between $|G|$ and $|\LCS|$. The normalizing
factor $n$ is the length of sequence 273614N.\label{table:G}}
\begin{tabular}{l@{\hspace{2ex}}|@{\hspace{3ex}}r@{\hspace{3ex}}r@{\hspace{3ex}}r@{\hspace{3ex}}r}
	&   322134S   & 378604X & BC187 & DBVPG1106 \\[.5ex]
\hline \\[-1.5ex]
$|\LCS|/n$  & 0.9341 & 0.9669 & 0.9521 & 0.9590\\[.5ex]
$|G|/n$     & 0.8694 & 0.8655 & 0.8798 & 0.8800\\[.5ex]
\end{tabular}
\end{table}

\section{Conclusions}\label{sec:concl}


In this paper we have considered the problem of building an index for a
string $S_2$ given an FM-index for a similar string $S_1$. We have shown how
to build such a ``relative'' index using space bounded by the BW-distance
between $S_1$ and $S_2$. The BW-distance is simply the edit distance between
$\BWT(S_1)$ and $\BWT(S_2)$ when only insertions and deletions are allowed.
We have also introduced the notion of BWT-invariant subsequence and shown
that it can be used to determine a set of $S_1$ suffix array samples that can
be easily ``reused'' for an index for $S_2$.

We have tested our approach by building a relative index for a Han Chinese
individual with respect to an FM-index of the human reference genome. We
leave as a future work the development of these ideas and the complete
implementation of a relative FM-index supporting locating and extracting.  We
also leave as future work proving bounds on the BW-distance and the length of
the longest BWT-invariant subsequence in terms of the edit distance of the
strings.

\bibliographystyle{plain}
\bibliography{relative}

\appendix

\section*{Appendix: Reusing an SA Sample} \label{sec:sample}

Consider the example strings $S_1$, $S_2$ given in the introduction.  The characters of
\(\BWT (S_1) [1..16]\) and \(\BWT (S_2) [1..15]\) are mapped to their
positions by the BWT from
\begin{align*}
& S_1 [16, 2, 6, 8, 13, 1, 12, 3, 7, 9, 14, 10, 15, 5, 11, 4]\\
& S_2 [15, 7, 2, 5, 12, 1, 11, 8, 3, 6, 13, 9, 14, 4, 10]
\end{align*}
respectively.  (Notice the lists of indices are just the SAs of $S_1\$$ and
$S_2\$$ with each value decremented.)  Therefore, if \(r = 3\) then
$$
R = 1000110010010001, \qquad
A [1..6] = [16, 13, 1, 7, 10, 4]\,.
$$
Comparing $R$ and
\(B_1 = 0001000000000111\)
we see that the sampled characters \(\BWT (S_1) [1, 5, 6, 9, 12]\) that are
in $C$, are $C$'s 1st, 4th, 5th, 8th and 11th characters.  From
\(B_2 = 010000000001010\)
we see that the 1st, 4th, 5th, 8th and 11th characters in $C$ in \(\BWT
(S_2)\) are \(\BWT (S_2) [1, 5, 6, 9, 13]\), which are mapped to their
positions by the BWT from \(S_2 [15, 12, 1, 3, 14]\).

The relative order \(5, 3, 1, 2, 4\) of the positions \(15, 12, 1, 3, 14\) in
$S_2$ of these characters, is {\em almost} the same as the relative order
\(5, 4, 1, 2. 3\) of the positions \(16, 13, 1, 7, 10\) in $S_1$ of the
sampled characters in \(\BWT (S_1)\) that are in $C$, which seems promising.
What if we choose $C$ and its occurrences in \(\BWT (S_1)\) and \(\BWT
(S_2)\) such that the relative order in $S_1$ of all \(\BWT (S_1)\)'s
characters that are in $C$, is the same as the relative order in $S_2$ of all
\(\BWT (S_2)\)'s characters that are in $C$?

For example, we can choose instead
\begin{align*}
C' & = \mathsf{TCTCGTAAAGG}\\
B_1' & = 0001000001010101 & B_2' & = 010000010001010\\
D_1' & = \mathsf{GAGTC} & D_2' & = \mathsf{GACC}
\end{align*}
even though $C'$ is not then an LCS of \(\BWT (S_1)\) and \(\BWT (S_2)\) and,
thus, our data structures for supporting \rank\ in \(\BWT (S_2)\) are
slightly larger.  With these choices, the characters in \(\BWT (S_1)\) and
\(\BWT (S_2)\) that are in $C'$, are mapped to their positions by the BWT
from
$$
S_1 [16, 2, 6, 13, 1, 12, 3, 7, 14, 15, 11],\qquad
S_2 [15, 2, 5, 12, 1, 11, 3, 6, 13, 14, 10]
$$
and the relative order \(11, 2, 4, 8, 1, 7, 3, 5, 9, 10, 6\) of the indices
in those two lists is the same, as desired.

Suppose we store yet another pair of bitvectors
$$
M_1 = 0001100111000000,\qquad
M_2 = 000100111000000
$$
with 1s marking the positions in $S_1$ and $S_2$ of characters that are not
mapped into $C'$ in \(\BWT (S_1)\) and \(\BWT (S_2)\).  We claim that if we
can support fast \rank\ queries on $B_2'$, $R$ and $M_1$, fast access to $A$
and fast $\select_0$ queries on $B_1'$ and $M_2$, then we can support fast
access to a (possibly irregular) sample SA sample for $S_2$ with as many
sampled characters as there are in $C'$ in \(\BWT (S_1)\).  More
specifically, if \(\BWT (S_2) [i]\) is in $C'$ and
\(R [B_1'.\select_0 (B_2'.\rank_0 (i))] = 1\)
--- meaning the corresponding character in $C'$ in \(\BWT (S_1)\) is sampled --- then \(\BWT (S_2) [i]\) is mapped to its position by the BWT from
\[S_2 \left[ \rule{0ex}{4ex}
    M_2.\select_0 \left( \rule{0ex}{3.5ex}
      M_1.\rank_0 \left( \rule{0ex}{3ex}
        A \left[ \rule{0ex}{2.5ex}
          R.\rank_1 \left( \rule{0ex}{2ex}
            B_1'.\select_0 \left(
              B_2'.\rank_0 (i)
            \right)
          \right)
        \right]
      \right)
    \right)
  \right]\,.\]
We leave a detailed explanation to the full version of this paper.  We note,
however, that this approach works for any sample rate $r$, and even if the SA
sample for $S_1$ is irregular itself.

In our example, since \(\BWT (S_2) [10]\) is in $C'$, \(B_1'.\select_0
(B_2'.\rank_0 (10)) = 9\) and \(R [9] = 1\), we know \(\BWT (S_2) [10]\) is
mapped to its position by the BWT from position
\(M_2.\select_0 \left(
    M_1.\rank_0 \left( \rule{0ex}{2ex}
      A [R.\rank_1 (9)]
    \right)
  \right)
= 6\)
in \(S_2\).

\end{document}